\documentclass[12pt]{article}

\usepackage[utf8]{inputenc}
\usepackage[T1]{fontenc}

\usepackage{needspace}
\usepackage{geometry}

\usepackage{amsthm, amsmath, amssymb}

\usepackage{float}
\usepackage{xcolor}
\usepackage{thmtools}
\usepackage{thm-restate}
\usepackage{enumerate}
\usepackage{hyperref}
\usepackage{graphicx}
\usepackage{complexity}

\usepackage{tabularx}
\usepackage{environ}
\makeatletter
\newcommand{\problemtitle}[1]{\gdef\@problemtitle{#1}}%
\newcommand{\problemparameter}[1]{\gdef\@problemparameter{#1}}
\newcommand{\probleminput}[1]{\gdef\@probleminput{#1}}
\newcommand{\problemquestion}[1]{\gdef\@problemquestion{#1}}
\NewEnviron{problemgen}{
\problemtitle{}\problemparameter{}\probleminput{}\problemquestion{}
\BODY
\par\addvspace{.5\baselineskip}
\noindent
\begin{tabularx}{\textwidth}{@{\hspace{\parindent}} l X c}
    \multicolumn{2}{@{\hspace{\parindent}}l}{\@problemtitle} \\
    \textbf{Input:} & \@probleminput \\
    \textbf{Output:} & \@problemquestion
  \end{tabularx}
  \par\addvspace{.5\baselineskip}
}

\usepackage[textsize=small, color=red!30]{todonotes}

\usepackage{lineno}

\usepackage{soulpos} 
\definecolor{lightblue}{rgb}{.88,.88,1}
\ulposdef{\revision}{} 

\hypersetup{
  pdftitle = {On the hardness of inclusion-wise minimal separators enumeration},
  pdfauthor = {Caroline Brosse, Oscar Defrain, Kazuhiro Kurita, Vincent Limouzy, Takeaki Uno, Kunihiro Wasa},
  colorlinks = true,
  linkcolor = black!30!red,
  citecolor = black!30!green
}

\newtheorem{theorem}{Theorem}[section]

\newtheorem{lemma}[theorem]{Lemma}
\newtheorem*{lemma*}{Lemma}
\newtheorem{corollary}[theorem]{Corollary}

\newcommand{\citep}[1]{\cite{#1}}
\newcommand{\myparagraph}[1]{\paragraph{#1.}}

\DeclareMathOperator{\var}{var}
\newcommand{\mes}{\textsc{Inclusion-wise Minimal Separators Enumeration}}

\def\cqedsymbol{\ifmmode$\lrcorner$\else{\unskip\nobreak\hfil
\penalty50\hskip1em\null\nobreak\hfil$\lrcorner$
\parfillskip=0pt\finalhyphendemerits=0\endgraf}\fi}

\begin{document}

\title{On the hardness of inclusion-wise minimal\\ separators enumeration\thanks{The first author acknowledges the support of the French government IDEX-ISITE initiative 16-IDEX-0001 (CAP 20-25) and the French Agence Nationale de la Recherche under contract Digraphs ANR-19-CE48-0013-01. The first and fourth author were partially supported by the ANR project GRALMECO (ANR-21-CE48-0004-01). The third author was partially supported by JSPS KAKENHI Grant Numbers JP21K17812, JP22H03549, JST CREST Grant Number JPMJCR18K3, and JST ACT-X Grant Number JPMJAX2105.}}
\author{%
Caroline Brosse\thanks{Universit\'e Clermont Auvergne, Clermont-Ferrand, France} \thanks{CNRS, Université Côte d'Azur, Inria, I3S, Sophia-Antipolis, France}\and
Oscar Defrain\thanks{Aix-Marseille Universit\'e, Marseille, France}\and
Kazuhiro Kurita\thanks{Nagoya University, Nagoya, Japan}\and
Vincent Limouzy\textsuperscript{$\dagger$}\and
Takeaki Uno\thanks{National Institute of Informatics, Tokyo, Japan}\and
Kunihiro Wasa\thanks{Hosei University, Koganei, Japan}}

\date{August 30, 2023} 

\maketitle

\vspace{-.5cm}
\begin{abstract}%
    Enumeration problems are often encountered as key subroutines in the exact computation of graph parameters such as chromatic number, treewidth, or treedepth.
In the case of treedepth computation, the enumeration of inclusion-wise minimal separators plays a crucial role.
However and quite surprisingly, the complexity status of this problem has not been settled since it has been posed as an open direction by Kloks and Kratsch in 1998. 
Recently at the PACE 2020 competition dedicated to treedepth computation, solvers have been circumventing that by listing all minimal \mbox{$a$-$b$} separators and filtering out those that are not inclusion-wise minimal, at the cost of efficiency.
Naturally, having an efficient algorithm for listing inclusion-wise minimal separators would drastically improve such practical algorithms.
In this note, however, we show that no efficient algorithm is to be expected from an output-sensitive perspective, namely, we prove that there is no output-polynomial time algorithm for inclusion-wise minimal separators enumeration unless $\P = \NP$.

    \vskip5pt\noindent{}{\bf Keywords:} output-sensitive enumeration, \NP-hardness, inclusion-wise minimal separators, minimal separators.
\end{abstract}

\section{Introduction}

In subsets enumeration problems, the goal is to output all subsets that satisfy a specified property.
Problems of that kind are often used to develop theoretical and practical algorithms for solving optimization problems.
Famous examples include maximal independent sets enumeration for the computation of the chromatic number in graphs \citep{lawler1976chromatic,eppstein2002small}, or the enumeration of minimal separators and potential maximal cliques for the computation of treewidth \citep{berry2000generating, bouchitte2002listing, fomin2008exact}.
Given the importance of treewidth, its exact computation has been the subject of several competitions including the 2016 and 2017 editions of the PACE challenge~\citep{Pace}. 
In that context as well, the enumeration approach has proved to be successful.
Specifically, the PACE 2017 winner's program was implemented with an algorithm based on potential maximal cliques enumeration and dynamic programming~\citep{DBLP:journals/jco/Tamaki19}, based on the algorithm by Bouchitt\'e and Todinca~\citep{bouchitte2002listing}.

Another example of the use of subsets enumeration toward exact computation is the computation of the related graph parameter of treedepth, which, in turn, was the subject of the PACE 2020 challenge~\citep{Pace}.
To compute treedepth, a recursive formula based on inclusion-wise minimal separators is known~\citep{Deogun:DAM:1999,Korhonen::2019}.
As in the case of PACE 2016 and 2017, several participants at PACE 2020 implemented programs with a combination of dynamic programming and enumeration, using this recursive formula.
More specifically, the second and third-prize algorithms are based on dynamic programming with respect to small and inclusion-wise minimal separators~\citep{Korhonen::2019,brokkelkamp2020pace,DBLP:conf/iwpec/Korhonen20a}.
An \emph{inclusion-wise minimal separator} in a graph $G$ is an inclusion-wise minimal subset of vertices disconnecting at least two vertices $a$ and $b$.
It is not to be confused with the slightly different notion of \emph{minimal separators} which are defined as inclusion-wise minimal subsets of vertices disconnecting a \emph{specific} pair of vertices $a,b$.
Indeed, minimal separators may contain other minimal separators as a subset.
We refer to Section~\ref{sec:preliminaries} for the definitions and a discussion on the differences between the two notions.
In both implementations~\citep{brokkelkamp2020pace,DBLP:conf/iwpec/Korhonen20a} the minimal separators are first computed by the algorithm of Berry et al.~\citep{berry2000generating}, and those that are not inclusion-wise minimal (or considered too large in~\citep{Korhonen::2019}) are then filtered out.
However, since the gap between minimal and inclusion-wise minimal separators may be exponential in the number of vertices, this approach has a huge impact on efficiency.
Hence, the problem of listing not all minimal separators but inclusion-wise minimal separators finds motivations in the quest of fast implementations for exact treedepth computation. 

However and quite surprisingly, the complexity status of this problem has not been settled since it has been posed as an open direction by Kloks and Kratsch in 1998~\citep{kloks1998listing}.
In this note, we show that unfortunately, no efficient algorithm is to be expected from an output-sensitive perspective.
Namely, we prove that there is no output-polynomial time algorithm for inclusion-wise minimal separators enumeration unless $\P = \NP$.

\section{Preliminaries}\label{sec:preliminaries}

\myparagraph{Enumeration}
When dealing with enumeration problems, the number of solutions can be large (typically exponential) with respect to the input size.
Therefore, we do not aim at algorithms running in polynomial time in the size of the input.
Instead, we need to take into account the potentially large number of solutions of the problem, and look for algorithms running in time polynomial in the input size plus the number of solutions.
Such kinds of algorithms are referred to as \emph{output-polynomial} time algorithms~\citep{Johnson:Yannakakis:Papadimitriou:IPL:1988}.
(We note that the output-polynomial time condition is sometimes stated as being polynomial in the sizes of the input plus the output, however that the two notions coincide here as the solutions we will consider are subsets of the ground set, hence that they are of polynomial size in the size of the input.) 
Then, a question of interest in enumeration is that of knowing if there exists an output-polynomial time algorithm for the considered problem.

\myparagraph{Graphs and separators}
Let $G = (V, E)$ be a graph on vertex set $V$ and edge set $E$.
In this paper, we only consider graphs with no loops nor parallel edges.
Let $v \in V$. 
We say that a vertex $u$ is \emph{adjacent to $v$} if $\{u,v\}\in E$, and that an edge $e\in E$ is \emph{incident to $v$} if $v\in e$.
The \emph{neighborhood} of $v$, denoted by $N(v)$, is the set of vertices that are adjacent to $v$ in $G$.
The \emph{degree} of $v$ is $|N(v)|$.
Let $S$ be a subset of vertices of $G$.
The graph \emph{induced by $S$}, denoted by $G[S]$, is the graph $(S, \{ e\in E : e\subseteq S\})$.
By $G-S$ we mean the graph $G[V\setminus S]$.
A \emph{path} in $G$ is a sequence $P=(v_1, \ldots, v_p)$ of distinct vertices such that $\{v_i, v_{i+1}\}\in E$ for any $1\leq i< p$.
It is called an \emph{$s$-$t$ path} if $v_1 = s$ and $v_k = t$.
A graph $G$ is \emph{connected} if for any pair of vertices $u, v \in V$ it contains an $u$-$v$ path.
In this paper, we suppose that $G$ is connected.

For two vertices $a$ and $b$, an \emph{$a$-$b$ separator} is a subset $S$ of vertices such that $a$ and $b$ are not contained in the same connected component in $G-S$.
It is called \emph{minimal} if no proper subset of $S$ is an \emph{$a$-$b$ separator}.
More generally, we say that a subset $S\subseteq V$ is a \emph{separator} if there exist $a,b\in V$ such that $S$ is an $a$-$b$ separator, and that it is a \emph{minimal separator} if there exist $a,b\in V$ such that $S$ is a minimal $a$-$b$ separator.
Note however that two minimal separators may be inclusion-wise comparable, a behavior that was observed in \citep{golumbic2004algorithmic}.
To see this, consider the \emph{banner graph} obtained from a cycle on four consecutive vertices $a,b,c,d$ by adding a pendant vertex $e$ adjacent to $d$, as shown in Figure~\ref{fig:banner}.
Then $\{b,d\}$ is a minimal $a$-$c$ separator, hence a minimal separator.
However, it contains $\{d\}$ as another minimal separator.
We say that a subset $S\subseteq V$ is an \emph{inclusion-wise minimal} separator if $S$ is a separator of $G$ that is minimal with that property.
Clearly, every inclusion-wise minimal separator of $G$ is a minimal separator, but the converse does not hold in general as the above example shows.
Testing whether a set $S$ is a minimal separator can be done in polynomial time~\citep{kloks1998listing,golumbic2004algorithmic}.
As of testing whether a separator $S$ is inclusion-wise minimal, we argue\footnote{Such an assertion is not direct as there exist graph $G$ and sets $S\subset S'\subset S''\subseteq V(G)$ with $S,S''$ being separators and $S'$ not; consider e.g.~a path on five vertices labeled $1,2,3,4,5$ in order, and subsets of vertices $S=\{2\}$, $S'=\{1,2\}$ and $S''=\{1,2,4\}$.} that it amounts to test whether $S\setminus \{u\}$ is not a separator for any $u\in S$, which can also be done in polynomial time.
The first implication follows by inclusion.
For the other direction, let us assume that $S\setminus \{u\}$ is not a separator for any $u\in S$. If $S$ is inclusion-wise minimal we are done, so suppose that there exists another separator $S'\subseteq S$ with $|S'|\leq |S|-2$.
Consider one such $S'$ maximal by inclusion and let $v\in S\setminus S'$.
As $S'\cup \{v\}$ is not a separator and $S'$ is, $S'$ separates $v$ from another vertex $w$ and we deduce $N(v)\subseteq S'$.
Thus $N[v]\subseteq S$ and since $S$ is a separator, so is $S\setminus \{v\}$ which contradicts the maximality of $S'$.

In this paper, we are interested in the following problem.

\begin{figure}
    \centering
    \includegraphics[scale=1.1]{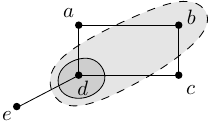}
    \caption{The banner graph. In this example, $\{b,d\}$ is a minimal $a$-$c$ separator but not an inclusion-wise minimal separator.}
    \label{fig:banner}
\end{figure}

\begin{problemgen}
  \problemtitle{\mes{}}
  \probleminput{A graph $G$.}
  \problemquestion{The set of inclusion-wise minimal separators of $G$.}
\end{problemgen}
As proven by Gaspers and Mackenzie~{\citep{GaspersM18}}, the number of inclusion-wise minimal separators of a graph may be exponential in its number of vertices.
To see this, consider the \emph{melon graph} 
on $3n + 2$ vertices obtained from $n$ disjoint paths on three vertices $\{u_i, v_i, w_i\}$, $1\leq i\leq n$ by adding an additional vertex $a$ adjacent to the $u_i$'s and an additional vertex $b$ adjacent to the $w_i$'s; see Figure~\ref{fig:melon} for an illustration.
In such a graph, every set in the family $\{\{x_1,\dots, x_n\} : x_i\in \{u_i, v_i, w_i\},\ 1\leq i\leq n\}$ defines an $a$-$b$ minimal separator, which is in fact an inclusion-wise minimal separator. 
As the number of such sets is $3^n$ the observation follows.

On the other hand, the number of minimal separators may be exponential in the number of inclusion-wise minimal separators.
To see this, consider a graph with exponentially many minimal separators in the number of vertices (e.g., a melon graph) to which we add a pendant neighbor to every vertex. 
The resulting graph has $O(n)$ inclusion-wise minimal separators (namely, each vertex defines such a separator) while the number of minimal separators has not decreased while adding the pendant vertices.
Consequently, listing inclusion-wise minimal separators from minimal separators using the algorithm in~\citep{berry2000generating} cannot yield a tractable algorithm.

\begin{figure}
    \centering
    \includegraphics[scale=1.1]{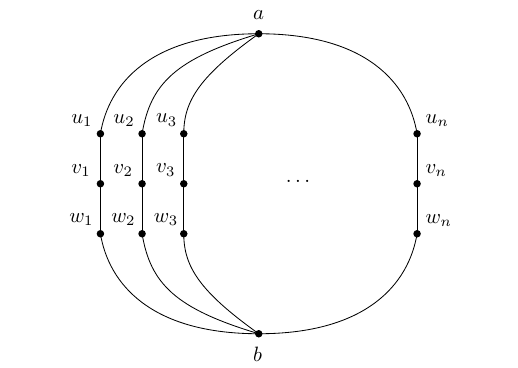}
    \caption{A melon graph on $3n+2$ vertices. Picking one vertex per set $\{u_i, v_i, w_i\}$ for every $1\leq i\leq n$ yields an inclusion-wise minimal separator, and there are $3^n$ such sets.}
    \label{fig:melon}    
\end{figure}

\section{Hardness of \mes{}}\label{sec:hardness}

In this section, we show that there is no output-polynomial time algorithm enumerating the inclusion-wise minimal separators of a graph unless $\P = \NP$.

Let $\phi = C_1 \land C_2 \land \dots \land C_m$, $m\geq 2$ be a $3$-CNF formula on $n$ variables and $m$ clauses.
We describe the construction of a graph $G=(V,E)$ on $O(n+m^2)$ vertices having an inclusion-wise minimal separator of size at least $4$ if and only if $\phi$ is satisfiable.
More specifically, we will show every such inclusion-wise minimal separator to be a minimal $a$-$b$ separator for two distinguished vertices $a$ and $b$.
The set $V$ is partitioned into four sets of vertices $V_1$, $V_2$, $V_3$ and $\{a, b\}$.
Intuitively, vertices in $V_1$ will represent the clauses in $\phi$, those in 
$V_2$ will guarantee consistency in variable assignments, and
$V_3$ will consist of pendant vertices that will be used to reduce the number of inclusion-wise minimal separators that do not separate $a$ and $b$.
The construction is detailed below and illustrated in Figure~\ref{fig:example}.

\begin{figure}
    \centering
    \includegraphics[scale=1.1]{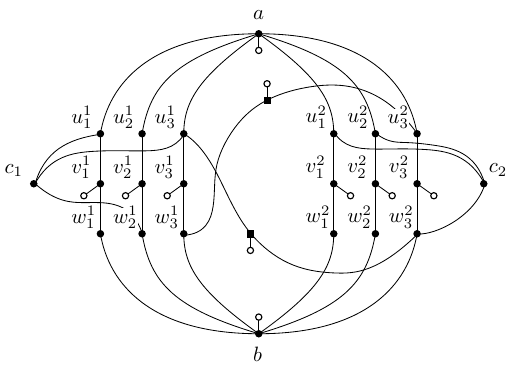}
    \caption{A representation of the graph $G$ associated to $\phi=(\bar{x}_1\vee x_2\vee \bar{x}_5)(\bar{x}_3\vee \bar{x}_4\vee x_5)$.
    Black squares represent vertices in $V_2$: they connect conflicting pairs of vertices arising from distinct clauses, here the vertices associated to the variable $x_5$. 
    White circles represent pendant vertices in $V_3$.}
    \label{fig:example}    
\end{figure}

Let us first define $V_1$ and parts of its incident edges.
For each clause $$C_j=(\ell_{1}^{j} \vee \ell_{2}^{j} \vee \ell_{3}^{j}),\ 1\leq j\leq m$$ in $\phi$, we create three induced paths $u_{1}^{j}v_{1}^{j}w_{1}^{j}$, $u_{2}^{j}v_{2}^{j}w_{2}^{j}$ and $u_{3}^{j}v_{3}^{j}w_{3}^{j}$ on three vertices each and add an extra vertex $c_j$.
We connect $a$ to each of $u_{1}^{j},u_{2}^{j},u_{3}^{j}$, and $b$ to each of $w_{1}^{j},w_{2}^{j},w_{3}^{j}$.
For every $i\in \{1, 2, 3\}$, we make $c_j$ adjacent to $u_{i}^{j}$ if $\ell_{i}^{j}$ is a negative literal, i.e., if $\ell_{i}^{j}=\bar{x}$ for some variable $x$ of $\phi$, and to $w_{i}^{j}$ otherwise.
In the assignments we will construct, selecting $u_{i}^{j}$ will in fact count for $x\mapsto 1$ while selecting $w_{i}^{j}$ will count for $x\mapsto 0$. 
Since $N(c_j)$ defines an inclusion-wise minimal separator of size three, these vertices $c_j$ will ensure that no inclusion-wise minimal separator of size greater than three contains $N(c_j)$ as a subset. 
Hence there will be at least one literal in $C_j$ that is assigned 1, and $C_j$ will be satisfied.
By construction, the obtained set of vertices $V_1$ consists of $10m$ vertices.
In the following, we shall call \emph{middle vertices} of the clause $C_j$ the vertices $v_{1}^{j}, v_{2}^{j}, v_{3}^{j}$ we created for~$C_j$.
Pairs of vertices \smash{$u_{i}^{j},w_{i'}^{j'}$} corresponding to the literals of a same variable, i.e., such that \smash{$\ell_{i}^{j}=\ell_{i'}^{j'}$ or $\ell_{i}^{j}=\bar{\ell}_{i'}^{j'}$} for non-necessarily distinct $i,i'\in \{1,2,3\}$ and $j,j'\in \{1,\dots,m\}$, will be referred to as \emph{conflicting} vertices.
We stress the fact that conflicting vertices are pairs of the form \smash{$u_{i}^{j},w_{i'}^{j'}$}, and that pairs of vertices \smash{$u_{i}^{j},u_{i'}^{j'}$} or \smash{$w_{i}^{j},w_{i'}^{j'}$} will not be considered conflicting even though they correspond to the literals of a same variable.

Now, we define $V_2$ and parts of its incident edges.
For every pair of conflicting vertices $u_{i}^{j},w_{i'}^{j'}$ with $j\neq j'$ we add a new vertex $y$ 
adjacent to \smash{$u_{i}^{j}$ and $w_{i'}^{j'}$}.
Since $N(y)$ defines an inclusion-wise minimal separator of size two, these vertices $y$ will ensure that no two conflicting literals are selected in an inclusion-wise minimal separator of size greater than two.
The set $V_2$ consists of these $y$ and hence has size $O(m^2)$.

Finally, let us define $V_3$ and the remaining edges of $G$.
The set $V_3$ consists of one pendant neighbor $z$ that is added to every vertex in $G-V_1$ as well as to every middle vertex in $V_1$.
Since $N(y)$ has size one, the role of $V_3$ will be to prevent inclusion-wise minimal separators of size greater than one from picking these vertices.
Then the size of $V_3$ is bounded by that of $V_1\cup V_2\cup \{a,b\}$, hence by $O(m^2)$.
This completes the description of $G$.

The proof that $\phi$ is satisfiable if and only if $G$ contains an inclusion-wise minimal separator of size at least $4$ is split into two lemmas.
First, let us prove that any inclusion-wise minimal separator of $G$ of size at least four 
implies a satisfying truth assignment for $\phi$.

\begin{lemma}\label{lem:sep}
    If $S$ is an inclusion-wise minimal separator of $G$ of size at least $4$, then $S\subseteq \{c_j, u_{i}^{j}, w_{i}^{j} : 1\leq i\leq 3,\ 1\leq j\leq m\}$ and $\phi$ is satisfiable.
\end{lemma}

\begin{proof}
    As $|S|\geq 4$, $S$ may not contain $N(x)$ as a subset for any vertex $x$ of degree less than $4$ in $G$.
    In particular, $S$ does not contain the neighbor of any pendant vertex in $G$ and clearly, as $S$ is minimal, it does not contain any pendant vertex neither.
    Since all elements in $V_2$ have pendant neighbors, we derive that $S\cap V_2=\emptyset$.
    For the same reason, neither $a$, $b$ nor any middle vertex may belong to $S$ and we conclude to the desired inclusion $S\subseteq \{c_j, u_{i}^{j}, w_{i}^{j} : 1\leq i\leq 3,\ 1\leq j\leq m\}$.

    We show that $S$ is in fact an $a$-$b$ separator.
    Let us consider two distinct connected components $A$ and $B$ of $G-S$ and suppose, toward a contradiction, that $a$ and $b$ belong to one such component.
    Say without loss of generality that $a,b\in A$.
    Let $u$ be a vertex of $B$ that is adjacent to $S$.
    Clearly, $u$ is not a neighbor of $a$ nor $b$.
    Thus $u$ either is a middle vertex, a vertex from $V_2$ or some $c_j$, $1\leq j\leq m$. 
    
    Now, note that every vertex $v$ that is a neighbor of $u$ and of one of $a,b$ must belong to $S$, since $u$ and $a,b$ belong to different connected components of $G-S$.
    We show that $S$ may not contain two conflicting vertices $u_{i}^{j}$ and $w_{i'}^{j'}$, $i,i'\in \{1,2,3\}$, $\smash{j,j'\in \{1,\dots, m\}}$.
    Indeed, we have two situations; (1) if $j=j'$ then $i=i'$ and $u_{i}^{j}$ and $w_{i'}^{j'}$ separate the middle vertex $v_{i}^{j}$ from the rest of the graph, or (2) if $j\neq j'$ then \smash{$u_{i}^{j}$ and $w_{i'}^{j'}$} separate a vertex from $V_2$ from the rest of the graph.
    In both situations, $S$ would contain as a subset another separator of size less than $4$.
    Therefore, $u$ is not a middle vertex, and neither does it belong to $V_2$.
    The only remaining possibility is that $u=c_j$ for some $1\leq j\leq m$.
    But as every neighbor of $c_j$ is a neighbor of $a$ or $b$, we deduce $N(c_j)\subseteq S$, which is also excluded as $S$ would contain as a subset of another separator of size less than $4$.
    We obtained the desired contradiction and conclude that $S$ is an $a$-$b$ separator.

    We are now ready to show that $\phi$ is satisfiable.
    Let $I$ be the assignment mapping variable $x_k$, $k\in \{1,\dots, n\}$ to $1$ if and only if $w_{i}^{j}\in S$ for some $1\leq i\leq 3,\ 1\leq j\leq m$.
    Since $S$ is an $a$-$b$ separator, every path $u_{i}^{j}v_{i}^{j}w_{i}^{j}$ of $G$ is intersected by $S$, and as previously shown, it is intersected on precisely one vertex that is not conflicting with other vertices in $S$.
    As we may assume that each variable appears in at least one clause in $\phi$, $I$ is a well-defined truth assignment.
    We note that $S$ may or may not contain vertices $c_j$, $1\leq j\leq m$, a point that is not relevant in what follows.
    However, since $N(c_j)\not\subseteq S$ for any $1\leq j\leq m$, for every clause $C_j$ there exists $i\in \{1,2,3\}$ and an endpoint $p$ of $u_{i}^{j}v_{i}^{j}w_{i}^{j}$ such that $p\in N(c_j)$, $p\not\in S$, and hence such that the other endpoint $q$ of $u_{i}^{j}v_{i}^{j}w_{i}^{j}$ belongs to $S$. 
    Thus the literal $\ell_{i}^{j}$ is evaluated to $1$ by $I$.
    We conclude that $I$ is a satisfying truth assignment of $\phi$ as desired.
\end{proof}

Conversely, we will show that each satisfying truth assignment of $\phi$ yields an inclusion-wise minimal separator of $G$ of size at least $4$.
To do this, for any satisfying truth assignment $I$ of $\phi$, we define two sets $T(I)$ and $F(I)$ as follows:
\begin{linenomath*}
\begin{align*}
    T(I):=\{u_{i}^{j} :\ &1\leq i\leq 3,\ 1\leq j\leq m\ \text{and}\ I(\var(\ell_{i}^{j}))=1\},\\
    F(I):=\{w_{i}^{j} :\ &1\leq i\leq 3,\ 1\leq j\leq m\ \text{and}\ I(\var(\ell_{i}^{j}))=0\},
\end{align*}
\end{linenomath*}
where $\var(\ell)$ denotes the variable corresponding to the literal and $I(x)$ denotes the valuation of variable $x$.

\begin{lemma}\label{lem:sat}
    If $I$ is a satisfying truth assignment of $\phi$ then there exists a set of integers $J\subseteq \{1,\dots,m\}$ such that $S=T(I)\cup F(I)\cup \{c_j : j\in J\}$ is an inclusion-wise minimal separator of $G$ of size at least $4$.
\end{lemma}

\begin{proof}
    In the following, we say that a clause $C$ is \emph{traversable}
    with respect to $I$ if $C$ contains at least one negative literal $\ell$ such that $I(\var(\ell))=0$, and at least one positive literal $\ell'$ such that $I(\var(\ell'))=1$, i.e., $C$ is satisfied by $I$ via both a positive and a negative literal.
    For example in Figure~\ref{fig:example}, the clause $(\bar{x}_1\vee x_2\vee \bar{x}_5)$ is traversable with respect to the assignment mapping every variable to $1$ except $x_1$ (which is mapped to $0$), but $(\bar{x}_3\vee \bar{x}_4\vee x_5)$ is not.
    Let $J$ be the set of indices of clauses of $\phi$ that are traversable with respect to $I$.
    We shall show that the lemma holds for such $J$.

    We first show that the described set $S$ is an $a$-$b$ separator of $G$.
    Let us assume toward a contradiction that this is not the case and let $P$ be an $a$-$b$ path in $G-S$.
    Since $I$ is a satisfying truth assignment, every path $u_{i}^{j}v_{i}^{j}w_{i}^{j}$ is intersected by $S$ and hence $P$ cannot follow a path $u_{i}^{j}v_{i}^{j}w_{i}^{j}$ for any $1\leq i\leq 3$ and $1\leq j\leq m$.
    Moreover, each vertex $y\in V_2$ has a neighbor in $S$ so $P$ does not reach any such $y$ as otherwise it would stop either at $y$ or at its pendant neighbor.
    The remaining case where $P$ contains $u_{i}^{j}$, then $c_j$, and then $w_{i'}^{j}$ for some $1\leq j\leq m$ and distinct $i,i'\in \{1,2,3\}$ corresponds to the situation of a traversable clause $C_j$, $j\in J$, and is excluded as $c_j\in S$ in that case.
    We conclude that $S$ is an $a$-$b$ separator of $G$ as desired.

    We note that $S$ is a minimal $a$-$b$ separator, as removing from $S$ any vertex belonging to $T(I)\cup F(I)$ or to $\{c_j : j\in J\}$ yields an $a$-$b$ path in $G-S$.
    Thus, for any proper subset $S'$ of $S$, the graph $G-S'$ contains a path connecting the endpoints of paths $u_{i}^{j}v_{i}^{j}w_{i}^{j}$ through $a$, $b$, and the $a$-$b$ path in $G-S'$.
    Thus an inclusion-wise minimal separator $S'$ of $G$ that is a proper subset of $S$, if it exists, must separate middle vertices, vertices $c_j$, $1\leq j\leq m$ or vertices from $V_2$ and $V_3$ from the component containing $a$ and $b$.
    Since the only way to separate these vertices using elements in $T(I)\cup F(I)\cup \{c_j : j\in J\}$ is to contain conflicting pairs of vertices or to contain $N(c_j)$ for some $j$, we conclude that no such $S'$ exists. 
\end{proof}

According to the previous two lemmas, the graph $G$ has an inclusion-wise minimal separator of size at least $4$ if and only if the formula $\phi$ is satisfiable.
Therefore, looking for an inclusion-wise minimal separator of size  $4$ or more is as hard as finding a truth assignment for $\phi$, as stated next.

\begin{corollary}
    Deciding if a graph $G$ has an inclusion-wise minimal separator of size at least $4$ is \NP-complete.
\end{corollary}

As we will see now, this result has even stronger implications from the enumeration point of view.
It is the point of the next theorem.

\begin{theorem}
    There is no output-polynomial time algorithm for \mes{} unless $\P = \NP$.
\end{theorem}

\begin{proof}
    Suppose for contradiction that there is an algorithm $\mathcal{A}$ enumerating the $d$ inclusion-wise minimal separators of an $n$-vertex and $m$-edge graph in total time that is polynomial in $n,m$ and $d$.
    Let $c\in \mathbb{N}$ be a constant such that the running time of $\mathcal{A}$ is bounded by $(n+m+d)^c$.
    We will prove that such an algorithm could be used to find a satisfying truth assignment for a $3$-SAT formula in polynomial time.
    
    Let $\phi$ be an instance of $3$-SAT on $N$ variables and $M$ clauses, and let $G$ be the $n$-vertex $m$-edge graph obtained from $\phi$ as described above.
    Then $n$ and $m$ are bounded by $O(N+M^2)$.
    We run $\mathcal{A}$ for $(n+m+d+1)^c$ steps on $G$ where $d=n^3$.
    If the algorithm has stopped within this time we check whether the obtained inclusion-wise minimal separators include one of size greater than $4$.
    If it is the case, then by Lemma~\ref{lem:sep} we conclude that $\phi$ is satisfiable.
    If not, we conclude that $\phi$ is not satisfiable by Lemma~\ref{lem:sat}.
    On the other hand, if the algorithm has not stopped then we conclude that the number of inclusion-wise minimal separators in $G$ is greater than $n^3$, hence that there are solutions of size at least $4$ in $G$.
    By Lemma~\ref{lem:sep} we may conclude that $\phi$ is satisfiable.
    The theorem follows observing that the procedure is polynomial in $N$ and $M$.
\end{proof}

\section{Conclusion}

We have shown that the enumeration of inclusion-wise minimal separators cannot be achieved in output-polynomial time unless $\P = \NP$, addressing an open question of Kloks and Kratsch in~\citep{kloks1998listing} and asserting that the algorithms in \citep{Korhonen::2019,brokkelkamp2020pace,DBLP:conf/iwpec/Korhonen20a} may not be improved by efficient output-sensitive enumeration algorithms for inclusion-wise minimal separators.
However, we note that our result does not give any insight on the existence of (input-sensitive) algorithms running in total time better than $2^n$, and that the existence of such algorithms would still benefit the algorithms in~\citep{Korhonen::2019,brokkelkamp2020pace,DBLP:conf/iwpec/Korhonen20a}.


\bibliographystyle{alpha}
\bibliography{main}

\end{document}